\documentclass[10pt]{article}

\usepackage{amssymb}
\usepackage{amsfonts}
\usepackage{amsmath}
\usepackage{amscd}
\usepackage[all,cmtip]{xy}
\usepackage{amsthm}
\usepackage{setspace}
\usepackage{enumerate}

\theoremstyle{plain}

\newtheorem{thm}{Theorem}

\newtheorem{cor}[subsection]{Corollary}

\theoremstyle{definition}

\newtheorem{defi}[subsection]{Definition}

\renewcommand{\AA}{{\mathbb A}}
\newcommand{\CC}{{\mathbb C}}
\newcommand{\PP}{{\mathbb P}}

\newcommand{\ZZ}{{\mathbb Z}}

\newcommand{\ol}{\overline}

\title{Duality of 2D gravity as a local Fourier duality}
\author{Martin T. Luu \footnote{Department of Mathematics, Stanford University, Stanford, CA 94305, USA, email: mluu@math.stanford.edu}}

\date{}

\begin{document}

\maketitle

\begin{abstract}
The p -- q duality is a relation between the $(p,q)$ model and the $(q,p)$ model of two-dimensional quantum gravity. Geometrically this duality corresponds to a relation between the two relevant points of the Sato Grassmannian. Kharchev and Marshakov have expressed such a relation in terms of matrix integrals. Some explicit formulas for small $p$ and $q$ have been given in the work of Fukuma-Kawai-Nakayama. Already in the duality between the $(2,3)$ model and the $(3,2)$ model the formulas are long. In this work a new approach to p -- q duality is given: It can be realized in a precise sense as a local Fourier duality of D-modules. This result is obtained as a special case of a local Fourier duality between irregular connections associated to Kac-Schwarz operators. Therefore, since these operators correspond to Virasoro constraints, this allows to view the p -- q duality as a consequence of the duality of the relevant Virasoro constraints.
\end{abstract}

\section{Introduction}
The local Fourier transform of $\ell$-adic sheaves was developed by Laumon in \cite{LAU} and applied in an arithmetic context. His work was inspired by Witten's proof of the Morse inequalities. The complex version of the local Fourier transform was defined by Bloch and Esnault in \cite{BE} and by Lopez in \cite{LOP}. We show here that this complex version is of importance in physics: 

Consider a pair of positive co-prime integers $(p,q)$. The partition function $\textrm{Z}(t_{1},t_{2},\cdots)$ of two-dimensional gravity coupled to a $(p,q)$ conformal field is given by 
$$\textrm{Z}(t_{1},t_{2},\cdots)=\tau^{2}(t_{1},t_{2},\cdots)$$
where $\tau$ is a certain $\tau$-function of the KP hierarchy. 
The p -- q duality of 2D quantum gravity, which is the analogue of T - duality for non-critical string theory, is a relation between the $(p,q)$ theory and the $(q,p)$ theory. Geometrically, this can be interpreted as relating the two points of the Sato Grassmannian that give rise to the relevant $\tau$-functions. We show that this duality can be described as a local Fourier transform.

In our approach this duality will be a special case of the following. Consider the $(W,Q)$ model of 2D gravity where $W$ and $Q$ are general polynomials, see \cite{KM} for details. It is described by a $W$-reduced $\tau$-function satisfying a suitable Virasoro constraint. This constraint can be formulated in terms of the action of the Kac-Schwarz operator $A^{W,Q}$ on the Sato Grassmannian. There is a natural way to associate an irregular connection on the formal punctured disc to these operators and we show that one can relate the connection of the $(W,Q)$ model with the connection of the $(Q,W)$ model via the local Fourier transform. In other words we show that there is a Fourier duality between Virasoro constraints. The p -- q duality is then obtained as a special case in the following manner: 

By the work of Schwarz \cite{SCH}, \cite{SCH2} one can attach D-modules to the $(p,q)$ models generalizing the results for the $(p,1)$ models due to Dijkgraaf-Hollands-Sulkowski \cite{DHS}. These D-modules are constructed from the quantization results of \cite{SCH}. The D-module of the $(p,q)$ model is related to the connection associated to the Kac-Schwarz operator and hence the duality of Virasoro constraints yields the p -- q duality. 

The reason that the local Fourier transform is relevant in describing the p -- q duality can be succinctly summarized. For two polynomials $W$ and $Q$ consider the Kac-Schwarz operator
$$A^{W,Q} := \frac{1}{W'(z)} \frac{\textrm{d}}{\textrm{d}z} - \frac{W''(z)}{2W'(z)^{2}}+Q(z).$$
It will turn out that to obtain a W -- Q duality one should relate $A^{W,Q}$ and its ``dual operator'' $\widehat{A}^{W,Q}$ and a crucial part in this will be played by expressing $W$ in terms of $Q$. In the local Fourier transform something very similar occurs: Given a meromorphic connection associated to a suitable Laurent series $f \in \CC(\!(\zeta^{1/p})\!)$ for some $p \ge 1$, the Fourier dual variable $\hat \zeta$ is determined via the equation
$$f(\zeta)=\frac{1}{\zeta \hat \zeta}$$
and one needs to use this to express $\zeta$ as a functions of $\hat \zeta$. It turns out that this inversion process in the case of connections associated to the Kac-Schwarz operators essentially corresponds to the inversion process of expressing $W$ in terms of $Q$.

\section{Two-dimensional quantum gravity}
In this section we recall basic aspects of the $(p,q)$ conformal field coupled to 2D gravity. It is known that this theory is described by a $\tau$-function of the KP hierarchy that is $p$-reduced and satisfies suitable Virasoro constraints. Such a $\tau$-function can be constructed via the Sato Grassmannian and one can translate the problem to finding suitable solutions to the string equation $[P,Q]=1$. We now recall this procedure and refer to \cite{FKN} for more details.  

For all the arguments the Sato Grassmannian is crucial, we hence recall the relevant definitions. Fix an indeterminate $z$ and let 
$\mathcal H=\CC(\!(1/z)\!)$. 
Then $\mathcal H =\mathcal H^{+} \oplus \mathcal H^{-}$ where $\mathcal H^{+}=\CC[z]$ and $\mathcal H^{-}=  \frac{1}{z}\CC[\![1/z]\!]$. Given a $\CC$-subspace $V$ of $\mathcal H$ let $\textrm{pr} : V \longrightarrow \mathcal H^{+}$ denote the projection map. The big cell of the Sato Grassmannian as a set is given by
$$Gr=\{ \CC\textrm{ - subspaces $V$ of } \mathcal H \textrm{ }\big | \textrm{ } \textrm{pr} \textrm{ is an isomorphism}   \}.$$  
Let $x$ be an indeterminate. By a differential operator we will mean an element of the complex Weyl algebra 
$$\textrm{D} := \CC [ x, \partial_{x}].$$
Differential operators act via a Fourier transform on $\mathcal H$:
$$x \mapsto - \partial_{z} \;\;\; \textrm{ and } \;\;\; \partial_{x} \mapsto z.$$ 
Let $\Psi$ denote the set of pseudo-differential operators:
$$\Psi :=\Big \{ \sum_{-\infty} a_{j}(x) \partial_{x}^{j} \textrm{ }\big | \textrm{ } a_{j}(x) \in \CC[\![x]\!] \Big \}.$$
One can reformulate, see \cite{FKN}, the conditions on the $\tau$-function of the $(p,q)$ model as finding a suitable point of the Sato Grassmannian stabilized by certain pseudo-differential operators. For example, the point of the $(p,q)$ model is stabilized by
$$\tilde P = \partial_{x}^{p}$$
$$\tilde Q = \frac{1}{p}\partial_{x}^{1-p}x +\frac{1-p}{2p}\partial_{x}^{-p} +\frac{1}{p} \sum_{i=1}^{p+q} i t_{i} \partial_{x}^{i-p}$$
with $t_{i}=0$ except for $t_{p+q}$. The Kac-Schwarz operator essentially corresponds to the $\tilde Q$-action in this formulation and this operator plays a crucial role in relating the $(p,q)$ model and the $(q,p)$ model via the local Fourier transform. 

The $(p,q)$ models have a generalization to the so-called $(W,Q)$ models where $W$ and $Q$ are polynomials of co-prime degree $p$ and $q$. In this formalism the $(p,q)$ model corresponds to the case $W=z^{p}$ and $Q=z^{q}$. In \cite{KM} a generalization of p -- q duality to W -- Q duality is obtained. A key feature of the work of Kharchev-Marshakov is that the known matrix integral representation of the $\tau$-function of the $(W,z)$ model turns into a duality relation of $\tau$-functions for the case of general $Q$. We describe this now in more detail.

\subsection{p -- q duality}
We briefly recall the basic set-up of the p -- q duality as described in the work of Kharchev and Marshakov \cite{KM}. In the generality of the $(W,Q)$ model this could be called the W -- Q duality. Let $W$ and $Q$ be single-variable polynomials of co-prime degree $p$ and $q$. We already recalled the definition of the Kac-Schwarz operator 
$$A^{W,Q}=\frac{1}{W'(z)} \frac{\textrm{d}}{\textrm{d}z} - \frac{W''(z)}{2W'(z)^{2}}+Q(z).$$
Kharchev-Marshakov define the dual operator as
$$\widehat A^{W,Q}= \frac{1}{Q'(z)} \frac{\textrm{d}}{\textrm{d}z} - \frac{Q''(z)}{2Q'(z)^{2}} - W(z).$$
In fact, what we denote here by $\widehat A^{W,Q}$ is simply called $A^{Q,W}$ in \cite{KM}. We have chosen a different notational convention in order to avoid assigning two interpretations to $A^{Q,W}$.

To obtain a $\tau$-function of the $(W,Q)$ model one should obtain a point $\mathcal V^{(W,Q)}$ of the Sato Grassmannian stabilized by $W$ and by $A^{W,Q}$. Kharchev-Marshakov consider the ``action'' 
$$S_{W,Q}(x,\mu)=Q(x)W(\mu)-\int_{0}^{x} W(y)Q'(y) \textrm{d} y.$$  
Their approach to the p -- q duality is to first write the point of the Grassmannian as
$$\mathcal V^{(W,Q)} = \textrm{Span}_{\CC}(\phi_{0},\phi_{1},\cdots)$$
with the $\phi_{i}$'s of the form
$$\phi_{i}(z) = W'(z)^{1/2}\exp(-S_{W,Q}|_{x=\mu=z}) \int f_{i}(x) \exp(S_{W,Q}(z,x)) Q'(x)^{1/2}  \textrm{d}x$$
for suitable functions $f_{i}$. It is then observed that the desired stabilization of the point $\mathcal V^{(W,Q)}$ by $W$ and by $A^{W,Q}$ translates into the stabilization of the point
$$\mathcal V^{(Q,W)} := \textrm{Span}_{\CC}(f_{0},f_{1},\cdots)$$
by $Q$ and $\widehat A^{W,Q}$. In this manner one obtains a version of the W -- Q duality. In the special case where $(W,Q)=(z^{p},z^{q})$ we call the point $\mathcal V^{(Q,W)}$ of the Sato Grassmannian and the corresponding $\tau$-function the $\widehat{(p,q)}$ model.
As shown in \cite{KM}, one can deduce from the above considerations a relation via matrix integrals between the two relevant $\tau$-functions that generalizes the matrix integral representation of the $\tau$-functions of topological models of 2D gravity: 

For an invertible $N\times N$ Hermitian matrix $M$ let 
$$t_{k} := \frac{\textrm{Trace } M^{-k}}{k}.$$
Then the $\tau$-function $\tau^{(W,Q)}$ of the $(W,Q)$ theory is given in terms of the $\tau$-function $\tau^{(Q,W)}$ of the $(Q,W)$ theory as an $N\times N$ matrix integral
$$ C[V,M] \cdot \int\tau^{(Q,W)}[X] \exp \left (\textrm{Trace} \left (\frac{1}{2} \log Q'(x) + \int_{X}^{M} W(z)Q'(z) +W'(X)Q'(M)  \right ) \right )  \textrm{dX }$$
where $C[V,M]$ is a certain Gaussian integral, see \cite{KM} for a definition. This matrix integral is in a certain sense independent of $N$ and depends on $M$ only through the times $t_{k}$. Even though the matrix integral representation of the W -- Q duality seems rather indirect, some explicit examples are known. Already for simple cases this can become quite complicated: 

For example in \cite{FKN} Fukuma-Kawai-Nakayama show by a different method that the $\tau$-function $\tau^{(2,3)}(x_{1},x_{3},x_{5})$ of the $(2,3)$ model is related to the $\tau$-function $\tau^{(3,2)}(y_{1},y_{2},y_{4},y_{5})$ of the $(3,2)$ model in the following manner: One defines the modified versions of the $\tau$-functions as
$$C^{2,3} := \log \tau^{(2,3)} - \frac{1}{25}\cdot \frac{x_{1}x_{3}^{3}}{x_{5}^{2}}+ \frac{3}{625}\cdot \frac{x_{3}^{5}}{x_{5}^{3}}+\frac{1}{10}\cdot \frac{x_{1}^{2}x_{3}}{x_{5}}+\frac{1}{40} \cdot \log x_{5}$$
and
\begin{eqnarray*}
C^{3,2} &:=& \log \tau^{(3,2)}+ \frac{2}{25} \cdot \frac{y_{1}^{2}y_{4}^{2}}{y_{5}^{2}}+\frac{1}{5}\cdot \frac{y_{1}y_{2}^{2}}{y_{5}} - \frac{16}{125} \cdot\frac{y_{1}y_{2}y_{4}^{3}}{y_{5}^{3}}+ \frac{128}{9375} \cdot \frac{y_{1} y_{4}^{6}}{y_{5}^{5}}- \frac{4}{75} \cdot\frac{y_{2}^{3} y_{4}}{y_{5}^{2}}+\frac{32}{625} \cdot \frac{y_{2}^{2}y_{4}^{4}}{y_{5}^{4}} \\
&& - \frac{512}{46875} \cdot \frac{y_{2} y_{4}^{7}}{y_{5}^{6}}+\frac{4096}{5859375} \cdot\frac{y_{4}^{10}}{y_{5}^{8}}+\frac{1}{15} \cdot \log y_{5}. 
\end{eqnarray*} 
Then an explicit relation between the two sets of variables $\{x_{1},x_{3},x_{5}\}$ and $\{y_{1},y_{2},y_{4},y_{5}\}$ is known and with this relation one obtains
$$C^{2,3}=C^{3,2} .$$

\subsection{D-modules for the $(p,q)$ models}
\label{D-modules-section}
Our approach to relate the two $\tau$-functions $\tau^{(W,Q)}$ and $\tau^{(Q,W)}$ is to consider relevant D-modules. Therefore, in this section we recall how the work of Schwarz associates D-modules to the $(p,q)$ models of 2D gravity, generalizing the construction for the $(p,1)$ case in \cite{DHS} done by Dijkgraaf-Hollands-Sulkowski.

\subsubsection{D-modules and $\tau$-functions}
Consider the complex Weyl algebra
$\textrm{D} = \CC [x, \partial_{x}]$. 
We start by discussing some aspects of the relation between D-modules and the $\tau$-functions of the $(p,q)$ models. 

Consider first the passage from suitable D-modules to $\tau$-functions of the KP hierarchy as discussed in \cite{DHS}. It is shown there that in a certain sense the $\textrm{D}$-module $\textrm{D}/\langle \partial_{x}^{2}-x \rangle$ is associated to the $(2,1)$ model. The main idea, see loc. cit. (Section 2.3), is to consider a suitable subspace of the space of solutions of the differential equation
$$(\partial_{x}^{2}-x)f(x)=0$$
and produce via the action of $\textrm{D}$ a point in the Sato Grassmannian and hence a $\tau$-function that should be the $\tau$-function associated to the $(2,1)$ model. The fact that solutions to the Airy equation are relevant for describing this model was already shown in the early 90's by Kac-Schwarz in \cite{KS}. The reason is that this differential equation can be related to the desired Virasoro constraints of the $\tau$-function. In the formulation via the Sato Grassmannian the constraints correspond to stabilization conditions of a point of the Grassmannian by certain operators and out of solutions of the Airy equation one can produce the basis vectors of such a special point. In order to really produce the desired point of the Grassmannian from the Airy function one has to make several conventions which we now explain: 

A problem is that the Airy function $\textrm{Ai}(x)$ is a solution to the above differential equation but does not yield an element of $\CC(\!(1/x)\!)$. Consider $y$ such that $y^{2}=x$. What is true, see for example \cite{KS}, is that 
$$\varphi := y^{1/2} \exp(2y^{3}/3) \cdot \textrm{Ai}(x) \in \CC(\!(1/y)\!).$$
It follows that after multiplication by the correction factor $y^{1/2} \exp(2y^{3}/3)$ the $\textrm{D}$-module $\textrm{D}/\langle \partial_{x}^{2}-x \rangle$ can be related to a $\tau$-function. Hence there are two questions:
\begin{enumerate}[(i)]
\item 
How to justify the multiplication by the correction factor.
\item
How to justify the change of variables from $x$ to $y$ such that $y^{2}=x$.
\end{enumerate}
We first recall the justification given in \cite{DHS} and then show a different viewpoint via Kac-Schwarz operators that we will then use for the case of general $p$ and $q$.

It is argued in \cite{DHS} that the exponential $\exp(2 y^{3}/3)$ can be discarded via the KP flow. Since in that reference sometimes the arguments mix the Segal-Wilson and the Sato version of the Grassmannian we give some more details: Let $\mathcal F$ denote the semi-infinite exterior power of the space $\mathcal H$. Define the time evolution operator as
$$\textrm{H}=\exp \left (\sum_{i=1}^{\infty} t_{i} z^{i} \right ) \in \CC(\!(1/z)\!) [\![t_{1},t_{2},\cdots]\!].$$
For a point $W$ of the big cell of the Grassmannian there exists a basis $\{v_{0},v_{1},\cdots \}$ for $W$ viewed as a $\CC$-vector space such that
$$v_{i}=z^{i}+ v_{i}^{-} \;\;\; \textrm{ with } \;\;\; v_{i}^{-} \in \mathcal H^{-}.$$
Let 
$$|W \rangle := v_{0} \wedge v_{1} \wedge \cdots \in \mathcal F.$$
For $z^{i}$ with $i \ge 0$ one can associate an element in the dual space of $\mathcal H$ and via the corresponding interior products one obtains a map
$$\langle 0 \big | - \rangle : \mathcal F \longrightarrow \CC.$$
See for example \cite{MUL} for details. The bosonization map
$$\mathcal F \longrightarrow \CC[\![t_{1},t_{2}\cdots,]\!]$$
satisfies
$$ |W\rangle \mapsto \langle 0  \; | \;  \textrm{H} v_{0} \wedge \textrm{H} v_{1} \wedge \cdots \rangle.$$ 
The KP $\tau$-function of a point $W$ of the Sato Grassmannian is known to be the image under the bosonization map of the corresponding point $|W\rangle$ of $\mathcal F$.
In this sense it is suggestive to interpret the exponential $\exp(2 y^{3}/3)$ as simply shifting the desired $\tau$-function with respect to the third KP time. In any case, we now show that getting rid of the correction factor as well as making the coordinate change $y=x^{2}$ can be re-interpreted in a very simple manner in terms of the Kac-Schwarz operators: 

For simplicity of notation define $A$ and $a$ via 
$$A:=A^{W=z^{p},Q=z^{q}} =\frac{1}{p z^{p-1}} \frac{\textrm{d}}{ \textrm{d} z} + a.$$
Then
$$\rho^{-1} A \rho = \frac{1}{p z^{p-1}} \frac{\textrm{d}}{ \textrm{d} z} =:\frac{\textrm{d}}{\textrm{d} z^{p}} \; \; \; \textrm{ for } \; \; \;\rho=z^{(p-1)/2}\exp(-\frac{p}{p+q} \cdot z^{p+q}).$$
Note that in the terminology of \cite{KM} this factor is essentially
$$W'(z)^{1/2} \exp(- S_{W,Q}|_{x=\mu=z}) $$ 
with $W=z^{p}$ and $Q=z^{q}$. Furthermore, the coordinate change $z^{p}=x$ is exactly what is needed in order to obtain $\rho^{-1} A \rho = \frac{\textrm{d}}{\textrm{d} x}$. Hence one can reinterpret the conventions made in \cite{DHS} in order to attach a D-module to the $(2,1)$ model as being the conventions necessary to turn the Kac-Schwarz operator into a meromorphic connection. 

We now show that with this convention, one can produce a D-module for the general $(p,q)$ model via the quantization scheme of \cite{SCH}.

\subsubsection{D-modules from quantization of differential operators}
\label{quantization-sub-section}
In \cite{SCH}, \cite{SCH2} Schwarz developed a notion of quantization of commuting differential operators. As part of that work there is an associated D-module. It has recently been shown by Liu-Schwarz that the quantization of the commuting pair of differential operators $\partial^{q}$ and $\partial^{p}$ produces the $\tau$-function of the $(p,q)$ model, see \cite{LS}, and in the $(p,1)$ case this produces the D-module attached to the $(p,1)$ model by Dijkgraaf-Hollands-Sulkowski. It then follows that the Schwarz quantization scheme yields a D-module associated to the general $(p,q)$ model. We now describe this in more detail.

Already for the $(2,1)$ model one can see easily the appearance of the so-called companion matrices: For the Kac-Schwarz operator
$$A:=\frac{1}{2z} \frac{\textrm{d}}{\textrm{d} z} -\frac{1}{4z^{2}} + \frac{3z}{2}$$
it is shown by Kac-Schwarz in \cite{KS} that the space
$\langle \varphi, A \varphi, z^{2} \varphi,\cdots \rangle$
yields a point in the big cell of the Grassmannian which yields the $\tau$-function of the $(2,1)$-model. 
Note that $A^{2} \varphi = \textrm{constant} \cdot z^{2} \varphi$ and hence one sees here through the $A$ action the occurrence of the matrix
$$\begin{bmatrix}
0&1 \\
z^{2} & 0
\end{bmatrix}$$
and this, after the coordinate change $z^{2}=x$, in fact can be viewed as the so-called companion matrix involved in the quantization of the differential operators $(\partial^{1},\partial^{2})$ in the sense of \cite{SCH}. We now recall this quantization scheme.

A differential operator $\sum_{i=0}^{n} a_{i}\partial^{i}$ is called normalized if $a_{n}=1$ and $a_{n-1}=0$.
If $Q_{0}$ is a normalized differential operator of degree $q$ then it is known that there exists  
$$S\in 1+\Psi^{(-1)}= \Big \{1 + \sum_{-\infty}^{-1} s_{j}(x) \partial_{x}^{j} \textrm{ }\big | \textrm{ } s_{j}(x) \in \CC[\![x]\!]\Big \}$$ 
such that 
$$SQ_{0}S^{-1}=\partial_{x}^{q}.$$
Let $V=S\mathcal H^{+} \in Gr$. For $0\le i \le q-1$ define
$$v_{i} = S z^{i} \in V.$$
Then every $v \in V$ can be written as 
$v=\sum c_{i}(z^{q}) v_{i}$ 
where $c_{i}(z^{q})$ is a polynomial in $z^{q}$. For a differential operator $P_{0}$ the companion matrix $M$ of $(P_{0},Q_{0})$ is defined via the action of $SP_{0}S^{-1}$ on the vectors $v_{i}$: 

It is the $q \times q$ matrix whose $(i+1,j+1)$'th entry for $0 \le i,j \le q-1$ is denoted by $M_{ij}$ and which is defined via
$$(SP_{0}S^{-1}) \cdot v_{i} = \sum_{j} M_{ij}(z^{q}) v_{j}$$
where $M_{ij}(z^{q})$ is a polynomial in $z^{q}$. For example, the main case of interest for us will be 
$$(P_{0},Q_{0})=(\partial^{q},\partial^{p})$$
with $p$ and $q$ co-prime. Write $q = sp + r$ with $0 \le r < p$. In this case the companion matrix is given by
$$M(p,q)=\begin{bmatrix} 
0 & & & 0 & z^{sp} &&& \\
 &\ddots & &&&\ddots &\\
0 & & &0 &&&z^{sp}\\
z^{(s+1)p} &  & &0 &0&&0\\
& \ddots & & &&\ddots&& \\
0& & & z^{(s+1)p} & 0 &&0
\end{bmatrix}$$ 
where there are $r$ columns containing a $z^{(s+1)p}$.

The notion of quantization developed in \cite{SCH} is based on fixing the companion matrix and letting the value of the commutator of two differential operators vary:
\begin{defi}[Schwarz]
Let $(P_{0}, Q_{0})$ be commuting differential operators with $Q_{0}$ normalized and companion matrix $M$. A quantization of $(P_{0},Q_{0})$ is a pair $(P,Q)$ of differential operators with $Q$ normalized such that 
$$[P,Q]=1$$ 
and such that the corresponding companion matrix equals $M$.
\end{defi}
For such a quantization it is known how to obtain an associated point of $Gr$ and hence a KP $\tau$-function. Let $p$ and $q$ be positive co-prime integers. It is shown by Liu-Schwarz in \cite{LS} that the quantization of $(\partial^{q},\partial^{p})$ yields the $\tau$-function of $(p,q)$ model of 2D quantum gravity. The reason that the $(p,q)$ model corresponds to the quantization of $(\partial^{q},\partial^{p})$ and not to the quantization of $(\partial^{p},\partial^{q})$ simply comes from different notational conventions. Note furthermore that the calculations in \cite{LS} are obtained for the $\hbar$-dependent string equation but this does not affect the following discussion.

In order to construct quantizations of commuting differential operators, in \cite{SCH} under suitable conditions a point $V$ of $Gr$ with $z^{p}$-basis $\{v_{0},\cdots,v_{p-1}\}$ is constructed such that
$$\big(\frac{1}{pz^{p-1}}\frac{\textrm{d}}{\textrm{d}z} + b(z) \big ) v_{i} = \sum M_{ij} v_{j}$$
for $b(z) \in \CC(\!(z)\!)$. To do so, for technical reasons one replaces the companion matrix $M$ by a gauge equivalent matrix $B$ given by
$$B_{ij}=M_{ij} z^{j-i}-\frac{i\delta_{ij}}{pz^{p}}.$$
The advantage is that under suitable assumptions the leading order coefficient matrix of $B$ has $p$ distinct eigenvalues and it follows from classical results obtained by Levelt and Turrittin that the connection 
$$\frac{\textrm{d}}{\textrm{d}x}-M(x)$$
can be transformed via gauge transformation to a connection with diagonal connection matrix. Here we set $x=z^{p}$. This diagonalization enables then to construct the desired point $V$ of the Sato Grassmannian.

In the case of the companion matrix associated to $(\partial^{q},\partial^{p})$ one can show that one of these diagonal entries is up to a normalization term given by
$$ z^{q} - \frac{p-1}{2p} \frac{1}{z^{p}}.$$
Since this is the crucial part of the Kac-Schwarz operator associated to the $(p,q)$ model one can deduce that the point of the Sato Grassmannian associated to the quantization has a KP $\tau$-function that satisfies the Virasoro constraints of the $(p,q)$ model, as desired. To sum up, consider the connection
$$\nabla_{p,q} := \frac{\textrm{d}}{\textrm{d}x} - M(p,q)(x) \; \; \; \textrm{ with } \; \; \; x:=z^{p} $$
on the formal punctured disc at $\infty$. Then one has:
\begin{thm}[Liu-Schwarz]
Let $p$ and $q$ be positive co-prime integers.
The $\textrm{\emph{D}}$-module $\nabla_{p,q}$ of the quantization of $(\partial^{q},\partial^{p})$ is associated to the $(p,q)$ model of 2D quantum gravity.
\end{thm}
Note that one could interpret $\nabla_{p,q}$ as a regular connection on the formal punctured disc centered at the point $0$ of the Riemann sphere. However, it follows from the above result that instead one should interpret them as  irregular connections on the formal punctured disc centered at $\infty$. 

Define the ``dual connection''
$$\widehat{\nabla}_{p,q} := \frac{\textrm{d}}{\textrm{d}x} + M(q,p)(x) \; \; \; \textrm{ with } \;\; \; x:=z^{q}.$$
Note that the companion matrix of $-\partial^{p}$ with respect to $\partial^{q}$ is simply $-M(q,p)$. One can deduce the following result which will be of importance for our considerations:
\begin{cor}
The $\textrm{\emph{D}}$-module $\widehat{\nabla}_{p,q} $ of the quantization of $(-\partial^{p},\partial^{q})$ corresponds to the $\widehat{(p,q)}$ model of 2D quantum gravity.
\end{cor}

\section{Local Fourier transform}
\label{Fourier-section}
The local Fourier transform has an interesting history. Originally, partially inspired by work of Witten, Laumon defined a local Fourier transform in the context of $\ell$-adic sheaves in order to obtain results about $\epsilon$-factors. This Fourier transform was then adapted to the context of complex connections by Esnault and Bloch and also Lopez. In the complex as well as in the $\ell$-adic case, the relation of the local Fourier transform to classical notions of Fourier transforms can be seen through their relations with suitable global Fourier transforms via a stationary phase principle. See \cite{LAU} and \cite{SAB}. 

We will follow the approach of Arinkin in describing the transform. It is developed in the context of $\PP^{1}$ and there exists several versions: $\mathcal F^{(x,\infty)}$, $\mathcal F^{(\infty,x)}$, $\mathcal F^{(\infty,\infty)}$ where $x$ is a point of $\PP^{1}$. For our purposes the last one is the most relevant and we now recall its definition. Note that the point $\infty$ in our applications will be the puncture of the Riemann Sphere involved in the Krichever construction.  

The local Fourier transform $\mathcal F^{(\infty,\infty)}$ concerns connections on the formal punctured disc at $\infty$. In the terminology of \cite{ARI} this corresponds to so-called holonomic $\textrm{D}_{K}$-modules where 
$$K=\CC (\!(\zeta)\!) \; \; \; \textrm{ and } \; \; \; \textrm{D}_{K}:=K[ \frac{\textrm{d}}{\textrm{d}\zeta}].$$
The category $\textrm{Hol}(\textrm{D}_{K})$ of holonomic $\textrm{D}_{K}$-modules is defined by Arinkin to be the full subcategory of the category of left $\textrm{D}_{K}$-modules which are finite dimensional as a $K$-vector space. In other words, one considers the category of finite-dimensional $K$-vector spaces with a meromorphic connection: 

For a finite-dimensional vector space $V$ over $\CC(\!(\zeta)\!)$ by a connection one simply means a $\CC$-linear endomorphism $\nabla$ that satisfies the Leibniz rule $\nabla(f v)=f \nabla(v)+( \frac{\textrm{d}}{\textrm{d} \zeta} f) v$ for all $f \in \CC(\!(\zeta)\!)$ and all $v \in V$. In terms of matrices such a connection is described by $$\nabla = \frac{\textrm{d}}{\textrm{d} \zeta} + B$$ for $B \in \mathfrak{gl}_{d}(\CC(\!(\zeta)\!))$ with $d = \dim_{\CC(\!(\zeta)\!)}(V)$. The morphisms are the morphisms of $K$-vector spaces that are compatible with the connections. 

Due to work of Levelt and Turrittin one has a classification of the category of holonomic $\textrm{D}_{K}$-modules. The formulation involves passing to a finite extension $\CC(\!(\zeta^{1/p})\!)$ of $\CC(\!(\zeta)\!)$ for suitable $p$. One of the reasons for passing to such an extension can be to have suitable eigenvalues of the connection matrix be contained in the field of scalars. See for example \cite{VAR} for historical context.

For $p\ge 1$ and $f \in \CC(\!(\zeta^{1/p})\!)$ denote by $E_{f,p}$ the $1$-dimensional vector space over $\CC(\!(\zeta^{1/p})\!)$ with the connection
$$\frac{\textrm{d}}{\textrm{d}\zeta} +\frac{f}{\zeta}.$$
The isomorphism class of $E_{f,p}$ depends only on the equivalence class of $f$ under the equivalence relation
$$f \sim \tilde f \; \; \; \;\; \textrm{ if } \; \; \; \;\;  f- \tilde f \in \zeta^{1/p} \CC[\![\zeta^{1/p}]\!] + \frac{1}{p} \ZZ.$$ 
Moreover, one can replace $f(\zeta)$ by $f(\mu_{p} \zeta)$ where $\mu_{p}$ is a $p$'th root of unity without changing the isomorphism class. Every holonomic $\textrm{D}_{K}$-module is built out of objects of the form $E_{f,p}$ in the following sense, see for example \cite{GRA} (Proposition 2.3): 

For $m_{i} \ge 1$ let $J_{m_{i}}$ be the object of $\textrm{Hol}(\textrm{D}_{K})$ given by the vector space $K^{m_{i}}$ with the connection
$$\nabla_{m_{i}}=\frac{\textrm{d}}{\textrm{d} \zeta} + \frac{1}{\zeta} \begin{bmatrix}
0 & 1 & &\\
&0&1 &\\
&& \ddots &1 \\
&&& 0
\end{bmatrix}$$
Then every object of $\textrm{Hol}(\textrm{D}_{K})$ is the direct sum of objects of the form
$$E_{f_{i},p_{i}} \otimes J_{m_{i}}$$
with $E_{f_{i},p_{i}}$ being irreducible, which is equivalent to the fact that $f_{i}$ is not an element of $\CC(\!(\zeta^{1/r})\!)$ for some $0<r <p_{i}$. Up to permutations this decomposition is unique. Note that for two objects $(M_{1},\nabla_{1})$ and $(M_{2},\nabla_{2})$ of $\textrm{Hol}(\textrm{D}_{K})$ their tensor product has underlying vector space $M_{1} \otimes_{\CC(\!(\zeta)\!)} M_{2}$ with connection given via
$$m_{1}\otimes m_{2} \mapsto \nabla_{1}(m_{1})\otimes m_{2} + m_{1} \otimes \nabla_{2}(m_{2})$$
for all $m_{1} \in M_{1}$ and $m_{2} \in M_{2}$.

The local Fourier transform is defined on a subcategory of the category of holonomic $\textrm{D}_{K}$-modules which we now describe. First, one defines
$$\textrm{slope} \left( E_{f,p} \right )= \textrm{Max}\left (0,- \textrm{ord}_{\zeta}(f)\right ).$$
Let $\textrm{Hol}(\textrm{D}_{K})^{>1}$ denote the full subcategory of $\textrm{Hol}(\textrm{D}_{K})$ of objects all of whose irreducible components have slopes bigger than $1$. The local Fourier transform is a functor
$$\mathcal F^{(\infty,\infty)} : \textrm{Hol}(\textrm{D}_{K})^{>1}\longrightarrow\textrm{Hol}(\textrm{D}_{K})^{>1}.$$
To describe it more precisely, let $j_{\infty}$ denote the inclusion of the punctured disc at infinity into $\AA^{1}$ and let $\mathcal F_{\textrm{glob}}$ denote the usual global Fourier transform of D-modules on the plane. Arinkin's approach to define the local Fourier transform is that it should make the following diagram commute:
$$\xymatrix{\textrm{Hol}(\textrm{D}_{K})^{>1} \ar[rr]^{j_{\infty,*}} \ar[dd]^{\mathcal F^{(\infty,\infty)}} & & \textrm{D} - \textrm{Mod} \ar[dd]^{\mathcal F_{\textrm{glob}}}\\ 
&&\\
\textrm{Hol}(\textrm{D}_{K})^{>1}  \ar[rr]^{j_{\infty,*}} &&\textrm{D} - \textrm{Mod}}$$
However, in general this does not determine the transform and continuity conditions are imposed. We refer to \cite{ARI} for details. It is important for us that one can give explicit equations to describe the local Fourier transform. We follow the exposition of Graham-Squire given in \cite{GRA}. 

For our purposes it suffices to describe the local Fourier transform of an irreducible $E_{f,p}$: Let $f \in \CC(\!(\zeta^{1/p})\!)$ be of order $-s/p$ with $s  > p$. Let $\hat \zeta$ denote the variable of the Fourier transform. It is defined via setting
$$f(\zeta)= \frac{1}{\zeta \hat \zeta}.$$
Since $f$ is of order less than $-1$, the Laurent series $\zeta f(\zeta)$ has negative order $(p-s)/p$ and it follows for example from \cite{GRA} (Lemma 5.1) that it has a formal compositional inverse in $\CC(\!(1/\zeta^{1/(s-p)})\!)$. Note that there is an error in the statement of the result in \cite{GRA}. One obtains an expression of $\zeta$ as a Laurent series with respect to $\hat \zeta^{1/(s-p)}$. Then one would like to define a new function $g \in \CC(\!(\hat \zeta^{1/(s-p)})\!)$ via
$$g(\hat \zeta):=-f(\hat \zeta)+\frac{s}{2(s-p)}$$
where $f(\hat \zeta)$ is obtained from $f(\zeta)$ by substituting for $\zeta$ its expression as an element of $\CC(\!(\hat \zeta)\!)$. A priori, the composition of the two Laruent series might not be well defined. However, without changing the isomorphism class of $E_{f}$ one can assume that $f \in \CC[\zeta_{p},\zeta_{p}^{-1}]$ and then this issue does not arise. With this assumption on $f$ the function $g$ is a well defined element of $\CC(\!(\hat \zeta^{1/(s-p)})\!)$ and it follows from \cite{GRA} (Theorem 3.6) that the Fourier transform is given by
$$\mathcal F^{(\infty,\infty)} \big ( E_{f,p} \big ) \cong E_{g,s-p}.$$

\subsection{Reformulation of $\mathcal F^{(\infty,\infty)}$}
\label{reformulation-section}
Due to notational conventions concerning Kac-Schwarz operators, it will be useful for our considerations to derive formulas for the local Fourier transform of elements of $\textrm{Hol}(\textrm{D}_{K})^{>1}$ if they are written in terms of the local coordinate at $0$ instead of $\infty$. To avoid notational confusion we denote this local Fourier transform by $\mathcal F^{\textrm{loc}}$ instead of $\mathcal F^{(\infty,\infty)}$.  

Let 
$$F = \sum_{t \le i \le s} a_{i} x^{i/p} \in \CC[x^{1/p},1/x^{1/p}]$$ 
for suitable $t, s$ with $a_{t}$ and $a_{s}$ non-zero and suppose that $F$ is not an element of $\CC(\!( 1/x^{1/r})\!)$ for some $0 < r <p$. Suppose further that $s  > p$. The connection 
$$ \frac{\textrm{d}}{\textrm{d} x}  + \frac{F(x)}{x}$$
in terms of $\zeta := 1/x$ is given as
$$ \frac{\textrm{d}}{\textrm{d} \zeta} -\frac{F(1/\zeta)}{\zeta} = E_{-F(1/\zeta),p}.$$
Hence, the variable $\hat \zeta$ of the Fourier transform satisfies
$$\frac{1}{\hat \zeta}=-\zeta \cdot F(1/\zeta).$$
The can be viewed as an element of $\CC(\!(\zeta^{1/p})\!)$. It follows for example from \cite{GRA} (Lemma 5.1) that this Laurent series has a compositional inverse $h(\zeta)$ in $\CC(\!(1/\zeta^{1/(s-p)})\!)$. It then follows that
$$\zeta=h(1/\hat \zeta) \in \CC(\!(\hat \zeta^{1/(s-p)})\!)$$
and
$$\frac{1}{h(\zeta)}= (-\frac{F(\zeta)}{\zeta})^{-1} \in \CC(\!(1/\zeta^{1/(s-p)})\!)$$
where the super-script $-1$ denotes the compositional inverse. For $\hat x = 1/\hat \zeta$ one obtains 
\begin{eqnarray*}
g(\hat \zeta)&=&\frac{s}{2(s-p)}+F(1/\zeta) \\
&=&\frac{s}{2(s-p)}-\frac{1}{\hat \zeta}\cdot \frac{1}{\zeta}\\
&=&\frac{s}{2(s-p)}- \frac{1}{\hat \zeta} \cdot \frac{1}{h(1/\hat \zeta)}\\
&=& \frac{s}{2(s-p)}-\frac{1}{\hat \zeta} \cdot \left (- \frac{F(x)}{x} \right)^{-1} ( \hat x).
\end{eqnarray*}
One obtains
\begin{eqnarray*}
\mathcal F^{\textrm{loc}} \left ( \frac{\textrm{d}}{\textrm{d} x} +\frac{F(x)}{x} \right ) &\cong & \frac{\textrm{d}}{\textrm{d} \hat x}-\frac{s}{2(s-p)} \cdot \frac{1}{\hat x} + \left (- \frac{F(x)}{x} \right)^{-1}(\hat x) \\
&\cong & \frac{\textrm{d}}{\textrm{d} \hat x}+\frac{s}{2(s-p)} \cdot \frac{1}{\hat x} +\left (- \frac{F(x)}{x} \right)^{-1} (\hat x).
\end{eqnarray*}
This reformulation of the local Fourier transform will be useful when we apply it to connections associated to Kac-Schwarz operators.

\section{Fourier duality of Kac-Schwarz connections and p -- q duality}
In this section we prove the main theorems. First we define the Kac-Schwarz connections and obtain a Fourier duality for them. We then deduce the p -- q duality as a special case.

Stabilization of points of the Sato Grassmannian by Kac-Schwarz operators corresponds to Virasoro constraints for the corresponding KP $\tau$-functions. We now describe a procedure for attaching irregular connections on the formal punctured disc to the Kac-Schwarz operators of the $(W,Q)$ model. Recall that one of the relevant operators is given by
$$A^{W,Q}=\frac{1}{W'(z)} \frac{\textrm{d}}{\textrm{d}z} - \frac{W''(z)}{2W'(z)^{2}}+Q(z).$$
It will be important for us to work with certain normalized versions of the Kac-Schwarz operators, essentially discussed by Kharchev-Marshakov in \cite{KM}: One makes the change of variables
$$\ol{\mu}^{p}=W(z) $$
and
$$\ol{\nu}^{q} = Q(z).$$
As noted in \cite{KM}, this yields equivalent solutions to the KP hierarchy. Let us now make the further change of variables $\ol{\mu}^{p} = x$. Then one can write the Kac-Schwarz operator as
$$\ol{A}^{W,Q}=\frac{\textrm{d}}{\textrm{d}x} -\frac{p-1}{2p} \frac{1}{x} +Q + \mathcal O(\frac{1}{x^{2}}).$$ 
If one interprets this as a connection at $\infty$ rather than at $0$ then the $\mathcal O(1/x^{2})$ contribution does not change the isomorphism class of the connection. This leads to:
\begin{defi}
The Kac-Schwarz connection of the $(W,Q)$ model is given by
$$\ol{A}^{W,Q}=\frac{\textrm{d}}{\textrm{d}x} -\frac{p-1}{2p} \frac{1}{x} +Q.$$ 
We will denote furthermore by $\widehat{\ol{A}}^{W,Q}$ the connection obtained in the analogous manner from the dual Kac-Schwarz operator $\widehat{A}^{W,Q}$.
\end{defi}

\begin{thm}
\label{Kac-Schwarz-theorem}
Let
$W$ 
and
$Q$ be polynomials of co-prime degrees $p$ and $q$. For $p$ odd one has
$$\mathcal F^{\textrm{\emph{loc}}} \left (\widehat{\ol{A}}^{Q,W} \right ) \cong \ol{A}^{Q,W}.$$
\end{thm}
\begin{proof}
Since $p$ is odd the term
$$\frac{1-p}{2p} \frac{1}{x}$$
of the Kac-Schwarz connection $\widehat{\ol{A}}^{Q,W}$ can be gauged away and therefore
$$\widehat{\ol{A}}^{Q,W}=\frac{\textrm{d}}{\textrm{d}x}  +\frac{F(x)}{x} \; \; \;\textrm{ with } \; \; \;
\frac{F(x)}{x}=-(Q \circ W^{-1})(x).$$
It follows that the order of $F$ is $-s/r$ where
$$s=p+q \; \; \; \textrm{ and } \; \; \; r=p.$$
In particular, $s>r$ and hence $\widehat{\ol{A}}^{Q,W}$ is an object in the category $\textrm{Hol}(\textrm{D}_{K})^{>1}$.
Moreover, since $p$ and $q$ are co-prime it follows that it is an irreducible object and we now calculate its local Fourier transform following the description in Section \ref{reformulation-section}: 

In the current situation one has 
$$h(x)=\left (-\frac{F(x)}{x}\right )^{-1}=(Q\circ W^{-1})^{-1}=W \circ Q^{-1}$$
and it follows from the calculation of $\mathcal F^{\textrm{loc}}$ in Section \ref{reformulation-section} that
\begin{eqnarray*}
\mathcal F^{\textrm{loc} } \left (\widehat{\ol{A}}^{Q,W} \right )
& \cong &
\frac{\textrm{d}}{\textrm{d}\hat x} +\frac{p+q}{2q} \frac{1}{\hat x}+(W \circ Q^{-1})(\hat x)\\
& \cong &
\frac{\textrm{d}}{\textrm{d}\hat x} +\frac{q-1}{2q} \frac{1}{\hat x}+(W \circ Q^{-1})(\hat x)\\
& \cong &
\frac{\textrm{d}}{\textrm{d}\hat x} +\frac{1-q}{2q} \frac{1}{\hat x}+(W \circ Q^{-1})(\hat x).
\end{eqnarray*}
Let us now compare this to the Kac-Schwarz connection $\ol{A}^{Q,W}$. Recall that for the corresponding Kac-Schwarz operator one has
$$A^{Q,W}=  \frac{1}{Q'(z)} \frac{\textrm{d}}{\textrm{d}z} - \frac{Q''(z)}{2Q'(z)^{2}} + W(z).$$
Furthermore, one has
$$\ol{A}^{Q,W}=\frac{\textrm{d}}{\textrm{d} \nu} -\frac{q-1}{2q} \frac{1}{\nu} +W.$$
Moreover, if $h$ is as before then 
$h(Q)=(W \circ Q^{-1} )(Q)=W$. 
It follows that 
$$  \ol{A}^{Q,W} \cong \frac{\textrm{d}}{\textrm{d} \nu}+\frac{1-q}{2q} \frac{1}{\nu} + h(\nu)\cong \mathcal F^{\textrm{loc}} \left ( \widehat{\ol{A}}^{Q,W}\right )$$ 
as desired.
\end{proof}
One can deduce the p -- q duality:
\begin{cor}
\label{main-theorem}
Let $p,q$ be positive co-prime integers. For $p$ odd one has
$$\mathcal F^{\textrm{\emph{loc}}} \Big ( \widehat{\nabla}^{q,p} \Big ) \cong  \nabla^{q,p}.$$
\end{cor}
\begin{proof}
It follows from the work of Liu-Schwarz \cite{LS}, as sketched in Section \ref{quantization-sub-section}, that
$$\nabla^{q,p}\cong \ol{A}^{z^{q},z^{p}}$$
and
$$\widehat{\nabla}^{q,p}\cong  \widehat{\ol{A}}^{z^{q},z^{p}}.$$
Hence, the desired result follows from Theorem \ref{Kac-Schwarz-theorem}.
\end{proof}
Note that since $p$ and $q$ are co-prime, always at least one of them will be odd. In this sense, the parity assumption of the previous results is not restrictive. Furthermore, one can actually obtain more general results without the parity assumption. We refer to \cite{LUS} for details.

\textbf{Acknowledgements:} 

I am deeply indebted to Albert Solomonovich Schwarz for numerous exchanges on the subject of this work and for sharing many crucial insights. In particular, it was his suggestion to attempt to use the local Fourier transform to explain some  Fourier duality in the quantization of differential operators that I had observed. Also, many thanks to A. Graham-Squire and the referee for helpful exchanges and remarks.

\end{document}